\documentclass[11pt,a4paper]{amsart}
\usepackage[top=4cm,bottom=4cm,left=3cm,right=3cm]{geometry}
\usepackage{amsmath,amssymb,amsfonts,amsthm}
\usepackage{graphicx}
\usepackage{textcomp}
\usepackage[dvipsnames]{xcolor}
\usepackage{url}

\usepackage[utf8]{inputenc}
\usepackage[T1]{fontenc}
\usepackage[nolist]{acronym}

\usepackage{amsmath, amssymb, amsfonts, amsthm}
\usepackage[all]{xy}
\usepackage{hyperref}
\usepackage{doi}
\usepackage{url}
\usepackage{mathtools}
\usepackage{calc}  
\usepackage{enumerate}
\usepackage{enumitem}
\usepackage{float}
\usepackage{bbm}
\usepackage{verbatim}
\usepackage{anyfontsize}

\usepackage{graphicx}
\graphicspath{{images/}{../images/}}
\usepackage{subcaption}
\usepackage{tabularx}
\usepackage{array}
\usepackage{wrapfig}
\usepackage{longtable}
\usepackage{makecell}
\usepackage{multirow, multicol}
\usepackage{ltablex}\keepXColumns

\newcolumntype{Y}{>{\centering\arraybackslash}X}
\newcolumntype{Z}{>{\scriptsize}Y}
\usepackage{rotating}

\usepackage{tikz}
\usetikzlibrary{shapes,arrows,positioning,arrows.meta,chains,positioning,quotes,fit}
\usepackage{blkarray}
\usepackage{tkz-euclide}
\usetikzlibrary{patterns} 
\usepackage{mathrsfs}
\usetikzlibrary{fadings}
\usetikzlibrary{shadows.blur}
\usetikzlibrary{spy}
\usetikzlibrary{decorations.pathreplacing}

\usepackage{nicematrix}

\usepackage{pgfplots}
\usepgfplotslibrary{colorbrewer}
\pgfplotsset{compat=1.17}

\usepackage{subfiles}
\usepackage{blindtext}
\usepackage{chngcntr}

    \usepackage{todonotes}
\usepackage{diagbox}
\usepackage{cleveref}
\usepackage{autonum}

\usepackage{centernot}

\usepackage{algorithm}
\usepackage{algpseudocode}

\usepackage[final]{pdfpages}

\usepackage{appendix}

\newtheorem{theorem}{Theorem}

\newtheorem{lemma}[theorem]{Lemma}
\newtheorem{proposition}[theorem]{Proposition}
\newtheorem{corollary}[theorem]{Corollary}

\theoremstyle{definition}

\newtheorem{definition}[theorem]{Definition}

\theoremstyle{remark}

\newtheorem{example}[theorem]{Example}
\newtheorem{conjecture}[theorem]{Conjecture}

\usepackage{natbib}
 
\def\tsc#1{\csdef{#1}{\textsc{\lowercase{#1}}\xspace}}
\tsc{WGM}
\tsc{QE}

\title{On Optimal Homogeneous-Metric Codes}

\begin{document}
\author[A. Pyka]{Andreas Pyka$^1$}
\author[V. Weger]{Violetta Weger$^1$}

\address{$^1$Technical University of Munich\\
        Germany 
}

\email{andreas.pyka@tum.de}
\email{violetta.weger@tum.de}

\maketitle

\begin{abstract}
The homogeneous metric can be viewed as a natural extension of the Hamming metric to finite chain rings. It distinguishes between three types of elements: zero, non-zero elements in the socle, and elements outside the socle. Since the Singleton bound is one of the most fundamental and widely studied bounds in classical coding theory, we investigate its analogue for codes over finite chain rings equipped with the homogeneous metric. We provide a complete characterization of Maximum Homogeneous Distance (MHD) codes, showing that MHD codes coincide with lifted MDS codes and are contained within the socle at low rank. Exceptions arise from exceptional MDS codes or single-parity-check codes. We then shift our focus to the Plotkin-type bound in the homogeneous metric and close an existing gap in the theory of constant homogeneous-weight codes by identifying those of minimal length. 
\end{abstract}

\maketitle

\section{Introduction}
The study of codes that are optimal with respect to classical bounds is a central theme in coding theory. 
In particular, codes attaining the Singleton bound, known as Maximum Distance Separable (MDS) codes, 
have been extensively investigated in the classical setting where codes are defined over finite fields 
endowed with the Hamming metric. In this case, several structural properties are well understood: 
MDS codes exist if $n\le q$, they are sparse for $n$ growing and  dense as $q$ grows, and their duals are again MDS codes. 
At the same time, important open problems remain, most notably the MDS conjecture.

Over the past decades it has become increasingly clear that many interesting code families 
are more naturally described over rings rather than fields. 
A striking example is provided by the Kerdock and Preparata codes, which are nonlinear over 
$\mathbb{F}_2$ but become linear over $\mathbb{Z}_4$ when endowed with the Lee weight via 
the Gray isometry, as studied in \cite{nechaev} and \cite{hammons}. 
This observation initiated the systematic study of ring-linear codes and revealed that 
suitable weight functions are required in order to capture the algebraic structure of the 
underlying ring. 
Among these, the homogeneous weight introduced by \cite{constantinescu} 
and later generalized by \cite{nechaevhonold} plays a prominent role, 
in particular for codes over finite Frobenius and chain rings.

A natural question is how classical bounds and the structure of optimal codes behave when 
one moves from the Hamming metric over fields to other metrics over rings. 
For example, in the Lee metric it was recently shown in \cite{LeePaper} that the situation 
differs drastically from the classical case: linear maximum Lee distance codes are extremely 
rare and, apart from a one-dimensional example over $\mathbb{Z}_5$, essentially do not exist.

In this paper we investigate the analogous questions for the homogeneous weight on finite 
chain rings. 
We first revisit Singleton-type bounds in this setting and introduce the class of 
\emph{Maximum Homogeneous Distance} (MHD) codes, i.e., codes attaining the homogeneous 
Singleton bound. 
Our first main contribution is a characterization of these codes, which allows us to 
analyze their density and compare their behavior to the classical case.

We then turn to Plotkin-type bounds for the homogeneous weight. 
Optimal codes for these bounds are constant-weight codes whose structure was largely 
characterized in \cite{wood}. 
In particular, Wood showed that every such code arises as a replication of a single 
shortest-length code. 
However, the construction of this minimal-length code was missing in the literature. 
We close this gap by explicitly constructing the smallest constant-weight homogeneous 
codes and thereby complete Wood's characterization. 
 
The paper is structured as follows: In Section \ref{sec:prelim} we give all the basic properties of finite chain rings, their linear codes and the homogeneous weight, which we require in the following sections. In Section \ref{sec:singleton} we state the generalized version of the Singleton bound, tailored to the homogeneous weight and  provide a characterization of optimal codes for this bound.  
In Section \ref{sec:plotkin} we provide an elementary proof of the  Plotkin bound in the homogeneous metric, derived from a classical proof over finite fields and thereby reobtain a chain ring specific case of a result by \cite{greferath}. We proceed by analyzing constant weight codes, building on the result from \cite{wood} that such codes are replications of a single shortest-length code. Finally, in Section \ref{sec:asymptotic} we take a look at asymptotic versions of both bounds and compare them to their classical counterparts. Part of these results first appeared in the bachelor's thesis of the first author~\cite{bachelor}.

\section{Preliminaries}\label{sec:prelim}

In this section we introduce the notation used throughout the paper and cover the basics of coding theory over finite chain rings. 

We denote by $\mathcal{R}$ a finite chain ring with maximal ideal generated by $\gamma$, let $s$ be its nilpotency index and $q=p^r$ be the size of the residue field, for a positive integer $r$. We further denote by $\mathcal{R}^\times$ its units. For $a \in \mathcal{R}$ we denote by $\langle a \rangle =a \mathcal{R}$ the (left)-ideal generated by $a.$  Since the notions of left and right chain rings coincide in the finite case, we will only speak of finite chain rings from now on.
We denote the identity matrix of size $n \times n$ by $I_n$ and for any matrix $A \in \mathcal{R}^{k \times n}$, we denote by $A^\top$ its transpose.

\begin{definition}
Let $x \in \mathcal{R}$. The largest positive integer $h$ with $h \leq s$ and $x \in \langle \gamma^h\rangle$ is called \textit{height} of $x$ and denoted by $\mathfrak{h}(x)= h$.
\end{definition}
Clearly, any element $x \in \mathcal{R}$ can be written  as $x=u \gamma^h$, where $h = \mathfrak{h}(x)$ and $u \in \mathcal{R}^\times$. 
While the unit $u$ is not unique, we have that $x=u \gamma^h=v \gamma^h$ for $u,v \in \mathcal{R}^\times$ if and only if $u-v \in \langle \gamma^{s-h} \rangle.$
We also recall that if $x,x' \in \mathcal{R}$ have heights $h,h'$ respectively, then $x+x'$ has height $\min\{h,h'\}.$
The socle of $\mathcal{R}$ is given by $\langle \gamma^{s-1} \rangle.$
Finally, we note that for all $i \in \{0,\ldots, s\}$, the ideal $\langle \gamma^i \rangle$ is isomorphic to $\mathcal{R}/\langle \gamma^{s-i}\rangle.$

\begin{definition}
An  $\mathcal{R}$\textit{-linear code} $\mathcal{C}$ of length $n$ is a (left-) $\mathcal{R}$-submodule of the free left module $\mathcal{R}^n$.
\end{definition}
In the remainder of the paper we will not differentiate between left- or right- modules, as 
 all of the following results can be  adapted to right $\mathcal{R}$-submodules of $\mathcal{R}^n$. The elements of $\mathcal{C}$ are called \textit{codewords}.

The fundamental theorem of finite abelian groups implies that any $\mathcal{R}$-module  is isomorphic to the following direct sum of $\mathcal{R}$-modules
\begin{align}
    \mathcal{C} \cong (\mathcal{R} / \gamma^s\mathcal{R})^{k_0} \times (\mathcal{R} / \gamma^{s-1}\mathcal{R})^{k_1} \times \cdots \times (\mathcal{R} /\gamma\mathcal{R})^{k_{s-1}} .
\end{align}
The unique $s$-tuple $(k_0, k_1, \ldots, k_{s-1})$ is called the \textit{subtype} of $\mathcal{C}$ and $k_0$ is called its \textit{free rank}.

Since $\mathcal{R}$ is finite, every code is finitely generated. The cardinality of a minimal  generating set  $K$ is called the \textit{rank} of $\mathcal{C}$  and can be computed as $K = \sum_{i = 0}^{s-1}k_i$. The $\mathcal{R}$-dimension of $\mathcal{C}$ is given by $k= \log_{|\mathcal{R}|} |\mathcal{C}| = \sum_{i=0}^{s-1} \frac{s-i}{s} k_i$. This allows for a matrix representation of the code.

\begin{definition}
   Let $\mathcal{C} \subseteq \mathcal{R}^n$ be a linear code of free rank $k_0$ and rank $K$. A matrix $G \in \mathcal{R}^{K \times n}$ is called a \textit{generator matrix} of $\mathcal{C}$ if the rows of $G$ span the code. A matrix $H\in \mathcal{R}^{(n-k_0)\times n}$  is called a \textit{parity-check matrix} of $\mathcal{C}$ if it has $\mathcal{C}$ as kernel.
\end{definition}
Usually, it is convenient to consider a generator and parity-check matrix in their standard form, as defined in \cite{norton}.

\begin{proposition}\label{sysform}
Let $\mathcal{C}$ be an $\mathcal{R}$-linear code of length $n$  and subtype $(k_0, \ldots, k_{s-1})$. Then $\mathcal{C}$ has (up to permutation of columns) a generator matrix in the following standard form 
\begin{align}
G = 
\begin{bmatrix} I_{k_0}&A_{0,1} &A_{0,2} &A_{0,3}&\dots& A_{0,s-1}& A_{0,s} \\
            0 &\gamma I_{k_1} & \gamma A_{1,2}& \gamma A_{1,3}& \dots& \gamma A_{1,s-1}&\gamma A_{1,s}\\
            0 &0 & \gamma^2 I_{k_2} & \gamma^2 A_{2,3}& \dots& \gamma^2 A_{2,s-1}&\gamma^2 A_{2,s}\\
            \vdots &  \vdots& \vdots & \vdots && \vdots &\vdots \\
            0 &0 & 0 & 0 & \dots & \gamma^{s-1}I_{k_{s-1}}& \gamma^{s-1} A_{s-1,s}\\
            \end{bmatrix}\ ,
\end{align}
where $A_{i,s}\in (\mathcal{R} / \gamma^{s-i} \mathcal{R})^{k_{i}\times (n-K)}$ and $A_{i,j}\in (\mathcal{R} / \gamma^{s-i} \mathcal{R})^{k_i\times k_j}$ for $j< s$. Moreover $ \mathcal{C} $ has (up to permutation of columns) a parity-check matrix in the following standard form  
        \begin{align} 
            H=\begin{bmatrix}
        B_{0,0}&B_{0,1}  &\dots& B_{0,s-1}& I_{n-K} \\
        \gamma B_{1,0}&\gamma B_{1,1} &\dots &\gamma I_{k_{s-1}}&0 \\
        \vdots &  \vdots&  & \vdots & \vdots \\
        \gamma^{s-1} B_{s-1,0} &\gamma^{s-1}I_{k_{1}}&  \dots& 0&0\\
        \end{bmatrix},
    \end{align} 
    where $B_{0,j} \in  (\mathcal{R} / \gamma^{s} \mathcal{R})^{(n-K)\times k_{j}}, B_{i,j} \in (\mathcal{R} / \gamma^{s-i} \mathcal{R})^{k_{s-i}\times k_{j}}, $ for $i>1.$ 
\end{proposition}
 
Let us consider the standard inner product, $\langle x,y \rangle=\sum_{i=1}^n x_iy_i$, which gives raise to the dual code. 
\begin{definition}
    Let $\mathcal{C} \subseteq \mathcal{R}^n$ be a linear code of subtype $(k_0, \ldots, k_{s-1})$ and rank $K.$ The dual code
    $\mathcal{C}^\perp = \{x \in \mathcal{R}^n \mid \langle x,c \rangle=0 \  \forall c \in \mathcal{C}\}$ 
    has subtype $(n-K, k_{s-1}, \ldots, k_1)$ and rank $n-k_0.$
\end{definition}

As in classical coding theory, we have that if $\mathcal{C}= \text{ker}(H^\top)$, then $\mathcal{C}^\perp= \langle H \rangle.$
The socle of a code $\mathcal{C}$ is given by the codewords of $\mathcal{C}$ which lie in $\langle \gamma^{s-1} \rangle$, that is $\mathcal{C}_0 = \mathcal{C} \cap \langle \gamma^{s-1} \rangle^n.$
The socle of a code connects codes over finite chain rings to codes over finite fields. This opens the door to apply results from  classical coding theory.

\begin{proposition}\label{socleisavs}
The socle $\mathcal{C}_0$ of a linear code $\mathcal{C} \subseteq \mathcal{R}^n$ is isomorphic to $\overline{\mathcal{C}} \subseteq (\mathcal{R}/\langle \gamma \rangle)^n = \mathbb{F}_{p^r}^n.$ 
\end{proposition}
This directly follows form the isomorphism from $\mathcal{R}/\langle \gamma^{s-1} \rangle \to  \mathbb{F}_{p^r}$ and as an immediate consequence we get that the socle $\mathcal{C}_0$ has $\mathcal{R}/\langle \gamma \rangle$-dimension $K.$

A classical code $\mathcal{C} \subseteq \mathbb{F}_q^n$ is said to be trivial if it has dimension $k \in \{0,1,n-1,n\}.$
We say that a ring-linear code $\mathcal{C} \subseteq \mathcal{R}^n$ is trivial, if its $\mathcal{R}$-dimension $k\in \{0,n\}$, which corresponds to the codes $\{0\}, \mathcal{R}^n.$ 
In fact, unlike in the classical case, we can have codes of $\mathcal{R}$-dimension 1, which do not behave trivially. For example $\mathcal{C}= \langle (2,0,0),(0,2,2) \rangle \subseteq \mathbb{Z}/4\mathbb{Z}^3$. The $\mathbb{Z}/4\mathbb{Z}$-dimension of $\mathcal{C}$ is 1, but its rank is $K=2.$

In this paper our main focus is on the homogeneous weight, however, we also require  the Hamming weight.  

\begin{definition}
For  tuples $x,y \in \mathcal{R}^n$ the \textit{Hamming weight} is
\begin{align}
\text{wt}_H(x) = |\{i \in \{1,\dots,n\} \mid  x_i \neq 0\}|.
\end{align}
This also induces the \textit{Hamming distance} between $x$ and $y$ as $d_H(x,y)=\text{wt}_H(x-y).$
If $\mathcal{C} \subseteq \mathcal{R}^n$ is a linear code, we define the \textit{minimum Hamming distance} of $\mathcal{C}$
\begin{align}
d_H(\mathcal{C}) = \min \{ \text{wt}_H(c) \mid c \in \mathcal{C} \setminus \{0\}\}.
\end{align}
\end{definition}

We will only consider the normalized homogeneous weight, but want to remark that all results carry over to the general definition of homogeneous weight.

\begin{definition}
For an element  $x \in \mathcal{R}$  the \textit{homogeneous weight} is
\begin{align} \text{wt}_{\text{Hom}}(x) = 
\begin{cases}
\frac{q}{q-1}, &\text{if } x \in \langle \gamma^{s-1} \rangle\setminus \{0\},\\
1, &\text{if } x \notin \langle \gamma^{s-1} \rangle,\\
0, &\text{if } x = 0.
\end{cases}
\end{align}
For a tuple $x \in \mathcal{R}^n$ we define its { homogeneous weight} as the sum of the homogeneous weight of its coordinates
\begin{align} \text{wt}_{\text{Hom}}(x) = \sum_{i=1}^n \text{wt}_{\text{Hom}}(x_i).
\end{align}
For $x,y \in \mathcal{R}^n$ we define the \textit{homogeneous distance} between $x$ and $y$ is given by $d_{\text{Hom}}(x,y)= \text{wt}_{\text{Hom}}(x-y).$
Finally, for a linear code $\mathcal{C} \subseteq \mathcal{R}^n$, we define \textit{the minimum homogeneous distance} of $\mathcal{C}$ as
\begin{align}
d_{\text{Hom}}(\mathcal{C}) = \min \{ \text{wt}_{\text{Hom}}(c) \mid c \in \mathcal{C} \setminus \{0\} \}.
\end{align}
\end{definition}

For $\mathcal{C} =\{0\}$,  in order to be in line with classical coding theory, we set $d_{\text{Hom}}(\{0\})= \frac{q}{q-1}(n+1)$.

If the finite chain ring $\mathcal{R}$ is fixed, we will often write $M$ for the maximal homogeneous weight of elements of $\mathcal{R}$, i.e., $M=\frac{q}{q-1}$.
Clearly, we have the following connection between the homogeneous and the Hamming weight:
for any $x \in \mathcal{R}^n$, we have 
$$\text{wt}_H(x) \leq \text{wt}_{\text{Hom}}(x) \leq M \text{wt}_H(x)$$ and for $x \in \langle \gamma^{s-1} \rangle^n$ we have that 
$$\text{wt}_{\text{Hom}}(x)=M \text{wt}_H(x).$$

The name of the homogeneous weight, comes from the following fact:
\begin{lemma}\label{avghom}
For any non-zero ideal $I \subseteq \mathcal{R}$ it holds
\begin{align}
|I| = \sum_{x \in I} \text{wt}_{\text{Hom}}(x).
\end{align}
\end{lemma}

For a linear code $\mathcal{C} \subseteq \mathbb{F}_q^n$ of dimension $k$, the classical Singleton bound states that 
\begin{align}
    d_H(\mathcal{C}) \leq n-k+1
\end{align}
and optimal codes with respect to this bound, i.e., codes which attain equality, are called Maximum Distance Separable (MDS). 

\cite{shiromoto} and \cite{samei} generalized this bound to more general weight functions on finite rings.
These weight functions encompass the homogeneous weight, the Lee weight and the Hamming weight. 
In our setting, their bounds state the following. 
\begin{theorem}\label{thm:homsbold}
    Let $\mathcal{C} \subseteq \mathcal{R}^n$ be a linear code of rank $K$. Then 
    $$ \left\lfloor \frac{d_{\text{Hom}}(\mathcal{C})-1 }{M} \right\rfloor \leq n-K.$$
\end{theorem}

We give a slightly improved version of the bound. To do so, we introduce an alternative floor function:

\begin{definition}
Set \(\left\lfloor \cdot \right\rfloor_h: \mathbb R \rightarrow \mathbb Z\) with
\[
\left\lfloor x \right\rfloor_h := 
\begin{cases} 
x - 1, & \text{for \(x \in \mathbb Z\),} \\
\left\lfloor x \right\rfloor, & \text{for \(x \notin \mathbb Z\).}
\end{cases}
\]
\end{definition}
 Clearly, $\left\lfloor x \right\rfloor_h \leq \left\lfloor x \right\rfloor.$

\begin{theorem}\label{singleton2}
Let  $\mathcal{C} \subseteq \mathcal{R}^n$ be a linear code of rank $K.$
Then 
\begin{align} 
d_{\text{Hom}}(\mathcal{C}) \leq M(n-K+1) \end{align} or equivalently 
\begin{align} \left\lfloor \frac{d_{\text{Hom}}(\mathcal{C})}{M} \right\rfloor_h \leq n - K.\end{align}
\end{theorem}

\begin{proof}
Since $\mathcal{C}_0 \subseteq \mathcal{C}$, it holds that 
$  d_{\text{Hom}}(\mathcal{C}) \leq d_{\text{Hom}}(\mathcal{C}_0) =Md_H(\mathcal{C}_0)$. Since 
$\mathcal{C}_0$ can be viewed as  a $K$-dimensional code $\overline{\mathcal{C}} \subseteq \mathbb{F}_q$, where $d_H(\overline{\mathcal{C}})=d_H(\mathcal{C}_0)$. Thus, we can apply the classical Singleton bound to get
\begin{align}
    d_{\text{Hom}}(\mathcal{C}) \leq d_{\text{Hom}}(\mathcal{C}_0)    \leq M(n-K+1). 
\end{align}

\end{proof}

\section{Maximum Homogeneous Distance Codes}\label{sec:singleton}

The two ways of phrasing the bound in \Cref{singleton2} do not lead to the same notion of optimality, since for the same code (e.g. \Cref{ex:tighterSing}), equality may be achieved in the second bound while the first bound stays an inequality. In the light of the following characterization, we choose the slightly less constrained option for our analysis.

\begin{definition}\label{MHDdef}
Let $\mathcal{C} \subseteq \mathcal{R}^n$ be a linear code of rank $K$ and minimum homogeneous distance $d$. We call $\mathcal{C}$ a \textit{Maximum Homogeneous Distance} (MHD) code if
\[
\left\lfloor \frac{d}{M} \right\rfloor_h = n - K
\]
or equivalently
\[
d > M(n-K).
\]
\end{definition}

We now want to take a look at a family of codes which are MHD. Our example will also justify the modification of the existing bound and show that the new versions are tighter.

\begin{example}\label{ex:tighterSing}
For $n,s$ positive integers and a prime $p \geq n$, take the code $\mathcal{C} \subseteq (\mathbb{Z}/ {p^s} \mathbb{Z})^n$ generated by
\[G = \begin{bmatrix}
1 & 0 & 1 & \dots & 1\\
0 & 1 & 2 & \dots & n-1
\end{bmatrix}.\]
We want to find its minimum homogeneous distance $d$. From the generator matrix, it is clear that for a non-zero linear combination of the rows, at most one coordinate $i>2$ can be zero, and only if the first two coordinates are non-zero. Also, one of the coordinates $i=1$ and $i=2$ always has to be non-zero, or else the codeword is zero. In total we get that $d=n-1$. With $M = \frac{p}{p-1}$ it holds
\begin{align}
d - M(n-K) &= (n-1)-\frac{p}{p-1}(n-2) = 1 + (n-2) - \frac{p}{p-1}(n-2) \\
&= 1 + \left(1-\frac{p}{p-1}\right)(n-2) = 1 - \frac{n-2}{p-1} > 0,
\end{align}
hence the code is MHD.  
Note that this code is not optimal with respect to Theorem \ref{thm:homsbold}:
\[
\left\lfloor\frac{d-1}{M}\right\rfloor = \left\lfloor\frac{n-2}{M}\right\rfloor < n - 2 =  n - K.
\]
We thus conclude that our modification led to a tighter bound.  
\end{example}

In the next step we will establish a relationship between MHD and MDS codes.

\begin{theorem}\label{mhdismds}
Let $\mathcal{C} \subseteq \mathcal{R}^n$ be an MHD code. Then $\mathcal{C}_0$ can be viewed as an MDS code $\overline{\mathcal{C}} \subseteq \mathbb{F}_q^n$.
\end{theorem}

\begin{proof} 
We have $d_{\text{Hom}}(\mathcal{C}_0) = M d_H(\mathcal{C}_0)$ and since $\mathcal{C}_0 \subseteq \mathcal{C}$, we get $$Md_H(\mathcal{C}_0) = d_{\text{Hom}}(\mathcal{C}_0) \geq d_{\text{Hom}}(\mathcal{C}) >M(n-K).$$
It follows that 
$d_H(\mathcal{C}_0) >n-K$. 
Since $\mathcal{C}_0$ can be identified with $\bar{\mathcal{C}} \subseteq \mathbb{F}_q^n$ of dimension $K$, we also need that $d_H(\mathcal{C}_0) \leq n-K+1$ and thus 
 $d_H(\mathcal{C}_0) = n-K+1$.
\end{proof}

We have thus established that all MHD codes are necessarily lifted MDS codes. Let us give another example of such a code.
\begin{example}
    Let us consider $\mathcal{R}= \mathbb{Z}/9\mathbb{Z}$ and the code $\mathcal{C} =\langle (1,1,2), (0,3,3) \rangle$. 
  We can easily check that $\mathcal{C}_0  =\langle (3,0,3), (0,3,3) \rangle,$  which is isomorphic to the MDS code $\overline{\mathcal{C}}= \langle(1,0,1),(0,1,1) \rangle \subseteq \mathbb{F}_3^3$. The minimum homogeneous distance of $\mathcal{C}$ is given by $d=3>M =3/2, $ hence $\mathcal{C}$ is MHD.
\end{example}

We can provide a condition on the residue field size of $\mathcal{R}$, which allows us to extend this result to an equivalence.

\begin{theorem}\label{mainchar}
Let $\mathcal{C} \subseteq \mathcal{R}^n$ be a linear code of rank $K$ with $q > n - K + 1$. Then $\mathcal{C}$ is an MHD code if and only if $\mathcal{C}_0$ is an MDS code.
\end{theorem}

\begin{proof}
One implication we have already seen in \Cref{mhdismds}. For the other direction we show the contraposition.
Note that for $a,b > 1$ it holds $a > b$ if and only if $\frac{a}{a-1} < \frac{b}{b-1}$, so our assumed inequality, i.e., $q>n-K+1$, can be written as \begin{align} \frac{q}{q-1}(n-K) < n-K+1.\end{align}
Now assume that $\mathcal{C}$ is not an MHD code. Then there exists a codeword $c \in  \mathcal{C} \setminus\{0\}$ with $\text{wt}_{\text{Hom}}(c) \leq M(n-K)$, where $M = \frac{q}{q-1}$. We can choose $g \in \{0, \ldots, s-1\}$ such that $\gamma^g c \in \mathcal{C}_0\setminus\{0\}$. For each $i \in \{1,\ldots,n\}$,   we have $\text{wt}_{\text{Hom}}(\gamma^g c_i) \leq M \text{wt}_{\text{Hom}}(c_i)$. In combination with the inequalities above, it holds
\[
\text{wt}_{\text{Hom}}(\gamma^g c) \leq M \text{wt}_{\text{Hom}}(c) \leq M^2(n-K) < M(n-K+1).
\]
Thus, the Hamming weight of $\gamma^g c$ satisfies $\text{wt}_H(\gamma^g c) = \frac{1}{M} \text{wt}_{\text{Hom}}(\gamma^g c) < n - K + 1$, i.e., $\mathcal{C}_0$ is not MDS.
\end{proof}

Next, we will see an example of an MHD code where \Cref{mainchar} can be applied.

\begin{example}\label{MHDex}
Let $\mathcal{R} = \mathbb Z/25\mathbb Z$ and let $\mathcal{C}$ be the code generated by
\begin{align} \begin{bmatrix}
1 & 1 & 1 & 1 & 1 & 0\\
0 & 5 & 10 & 15 & 20 & 0\\
0 & 5 & 20 & 20 & 5 & 5
\end{bmatrix}\end{align}
of length 6 and subtype $(1,2)$. The socle $\mathcal{C}_0$, seen over $\mathbb F_5$, is generated by 
\begin{align} \begin{bmatrix}
1 & 1 & 1 & 1 & 1 & 0\\
0 & 1 & 2 & 3 & 4 & 0\\
0 & 1 & 4 & 4 & 1 & 1
\end{bmatrix},\end{align}
which is the standard form of an extended Reed-Solomon code, i.e., $C_0$ is MDS. \Cref{mainchar} now implies that $\mathcal{C}$ is MHD. 
\end{example}

At first glance, \Cref{mainchar} may seem like a partial result only. But under assumption of the MDS conjecture it already covers almost all cases and the remaining ones can be characterized individually.

\begin{conjecture}[MDS Conjecture]
If there is a non-trivial MDS code $\mathcal{C} \subseteq \mathbb{F}_q^n$ of dimension $k$, then $n \leq q+1$, except when $q$ is even and $k = 3$ or $k = q-1$ in which case $n \leq q+2$.
\end{conjecture}
Note that the conjecture has been proven for prime $q$ in \cite{simeon}.

Assuming the conjecture is true, it follows from \Cref{mhdismds} that if $n>q+1$, then an MHD code has a trivial socle, i.e., $K \in \{0,1,n-1,n\}.$ An exception is the case where $q$ is even and $K = 3$ or $K = q-1$ where an MHD code with $n = q+2$ may exist. 
 Let us first consider the case $n\leq q+1,$ before we give a full characterization of all MHD codes. Note that if $q\leq n-K+1$, we have
$$n \leq   q+1 \leq n-K+2$$ 
which implies $K \leq 2$. Thus, for $n=q+1,$ we first consider $K=2$ separately.

\begin{theorem}\label{Keqtwo} Let $n = q + 1$.
Let $\mathcal{C} \subseteq \mathcal{R}^n$ be an MHD code of rank $K = 2$. Then $\mathcal{C}  \subseteq \langle \gamma^{s-1}\rangle^n$, that is the code lives in the socle.
\end{theorem}
\begin{proof}
By \Cref{MHDdef}, it holds $d_{\text{Hom}}(\mathcal{C})> q$. Assume that there exists $c \in \mathcal{C} \setminus \mathcal{C}_0$ and write $t := \min \{h \mid \exists i \text{ such that } \mathfrak{h}(c_i) = h\}$, so $t < s-1$ by assumption. 

We have $\gamma^{s-1-t}c \in \mathcal{C}_0 \setminus \{0\}$ and thus $\text{wt}_H(\gamma^{s-1-t}c) \geq d_H(\mathcal{C}_0) = n - 1$, since $\mathcal{C}_0$ is MDS by \Cref{mhdismds}. It follows by the minimality of $t$, that there exists a set $S \subseteq \{1,\dots,n\}$ of size $|S| \geq n-1$ with $\mathfrak{h}(c_i) = t$ for all $i \in S$. The case $|S| = n$ leads to a contradiction to the Plotkin bound. In fact, recall that the Plotkin bound states that $$d_H(\mathcal{C}_0) \leq \frac{nq^{K-1}(q-1)}{q^K-1}.$$  Since $d_H(\mathcal{C}_0) = q$ and $K = 2$, we have that  $\mathcal{C}_0$ meets the Plotkin bound with equality. Thus, the socle $\mathcal{C}_0$ has constant Hamming weight $n - 1$, whereas $|S| = n$ would imply $\text{wt}_H(\gamma^{s-1-t}c) = n$.

Hence, we have $|S| = n - 1$. Let $j$ be the only index in $\{1, \ldots, n\}\setminus S$ and note that $c_j = 0$  is impossible, since it would imply $\text{wt}_{\text{Hom}}(c) = \text{wt}_H(c) = n - 1 < d_{\text{Hom}}(\mathcal{C})$. Now define $c' = \gamma^{s-1-\mathfrak{h}(c_j)} c$, so $c'_j \in \langle \gamma^{s-1}\rangle \setminus \{0\}$. Since  $t$ was chosen minimal and $|S|\neq n$, we get that $\mathfrak{h}(c_j)>t$ and hence $c'_i \in \mathcal{R} \setminus \langle \gamma^{s-1}\rangle$ for all $i \in S$. Since $\mathcal{C}_0$ is an MDS code, it must contain $x \in\mathcal{C}_0$ with 
$x_j \neq 0$. As $\langle \gamma^{s-1} \rangle \setminus \{0\} = \mathcal{R}^\times \gamma^{s-1}$ there exists a unit $u \in \mathcal{R}^\times$ such that $u c'_j + x_j = 0$.   
For all $i \in S$ it holds $\mathfrak{h}(uc'_i) = \mathfrak{h}(c'_i) < s-1$ and $\mathfrak{h}(x_i) \geq s-1$, thus $\mathfrak{h}(uc'_i + x_i) < s-1$, i.e., $\text{wt}_{\text{Hom}}(uc'_i + x_i) = 1$.
In total, we get $\text{wt}_{\text{Hom}}(uc' + x) = q < d_{\text{Hom}}(\mathcal{C})$, which is a contradiction.
\end{proof}

It shall be noted here that the result from Theorem \ref{Keqtwo} is in fact an equivalence, since any MDS code $\mathcal{C}  \subseteq \langle \gamma^{s-1}\rangle^n$ is automatically MHD.

We are now ready to move to the full characterization of MHD codes.

\begin{theorem}\label{thm:fullchar}
    Assuming the MDS conjecture, the following conditions under their respective parameters are equivalent to $\mathcal{C} \subseteq \mathcal{R}^n$ being an MHD code
    \begin{enumerate}
        \item all trivial codes are MHD, i.e. $K \in \{0,n\}$,
        \item $n >q+1$ and $K=1$ with $\mathcal{C} \subseteq \langle \gamma^{s-1} \rangle^n$ is an MDS code,
        \item $n>q+1$, $q>2$ and $K=n-1$ with $\mathcal{C} \cap \langle \gamma^{s-1} \rangle^n$ is an MDS code,
        \item $n=q+1$ and $K \in \{1,2\}$ with $\mathcal{C} \subseteq \langle \gamma^{s-1}\rangle^n$ is an MDS code,
        \item $n=q$ and $K = 1$ with $\mathcal{C} \subseteq \langle \gamma^{s-1}\rangle^n$ is an MDS code,
        \item $n \leq q+1$ and $K>n-q+1$ with $\mathcal{C} \cap \langle \gamma^{s-1} \rangle^n$ is an MDS code,
    \end{enumerate}
    Additional MHD codes exist in the following cases with $\mathcal{C} \cap \langle \gamma^{s-1} \rangle^n$ being MDS as a necessary condition:
    \begin{itemize}
        \item $n > q+1$, $q=2$ and $K=n-1$, (so $\mathcal{C} \cap \langle \gamma^{s-1} \rangle^n$ is a single-parity-check code),
        \item $q$ even, $n=q+2$ and $K \in \{3, n-3\}$, (so $\mathcal{C} \cap \langle \gamma^{s-1} \rangle^n$ is an exceptional MDS code).
    \end{itemize}
\end{theorem}

Thus, except for these two special cases, either all lifted MDS codes are MHD or only those which live in the socle, where the distinction depends solely on the base parameters.

\begin{proof}
Recall that due to the MDS conjecture and Theorem \ref{mhdismds}, we get that if $n > q+1$, then $K \in \{0,1,n-1,n\}.$ An exception is the case where $q$ is even and $K = 3$ or $K = q-1$ where an MHD code with $n = q+2$ exists, as Proposition \ref{counterex1} shows and this covers the last additional case.

\begin{enumerate} 
\item For $K=0$, as we defined $d_{\text{Hom}}(\{0\})= M(n+1)$, we get that $\{0\}$ is an MHD code. Similarly, for $K=n$, the definition of MHD requires $d_{\text{Hom}}(\mathcal{C}) >M(n-K)=0$, which immediately leads to any code of rank $n$ being MHD.

\item \label{Keqone} In the case $K=1,$ we have a code $\mathcal{C} = \langle g \rangle \subseteq \mathcal{R}^n.$ If $\mathcal{C}$ is non-degenerate and  has rank $K=k_i$, i.e., $g_j \in \langle \gamma^i \rangle$ for all $j \in \{1, \ldots, n\}$, we can still have a degenerate socle $\mathcal{C}_0 = \langle \gamma^{s-i-1} g \rangle$. By Theorem   \ref{mhdismds}, however, we must have $d_H(\mathcal{C}_0)=n,$ hence $\mathfrak{h}(g_j)=i$ for all $j \in \{1, \ldots, n\}$ and thus $\mathcal{C}$ is not MHD if $K \neq k_{s-1}$ as $d_{\text{Hom}}(\mathcal{C})= n \leq M(n-1)$. Thus, for any $n \geq q$, if $K=1$,  all MHD codes are given by $\mathcal{C} = \langle g \rangle$ with $\mathfrak{h}(g_j)=s-1$ for all $j \in \{1, \ldots, n\},$ i.e., they live solely in the socle.

\item Here we have $q > 2 = n - K + 1$, so Theorem \ref{mainchar} can be applied.

\item This case is covered by Theorem \ref{Keqtwo} and the argument for case \ref{Keqone} above.

\item This case is covered by the argument for case \ref{Keqone} above.

\item Follows from Theorem \ref{mainchar}.

\end{enumerate}
Together with the two additional cases where only Theorem \ref{mhdismds} holds, this covers all possible parameters.

\end{proof}

We now give examples which show that the two special cases in Theorem \ref{thm:fullchar} do not follow the same pattern as the other cases, i.e., neither do all MHD codes live in the socle nor is the MDS property sufficient for a lifted code to be MHD. Two of the following codes are based on the smallest exceptional MDS code from \cite{hurley}:

\begin{lemma}\label{exceptionalMDS}
Write $\mathbb F_4 = \mathbb F_2(\alpha)$ with $\alpha^2 = \alpha + 1$. The code  $\mathcal{C} \subseteq \mathbb F_4^6$ generated by 
\begin{align} G =
\begin{bmatrix}
1 & 1 & 1 & 1 & 0 & 0\\
0 & 1 & \alpha & (\alpha + 1) & 1 & 0\\
0 & 1 & (\alpha + 1) & \alpha & 0 & 1
\end{bmatrix}
\end{align} is an MDS code.
\end{lemma}

\begin{proposition}\label{counterex1}
Let $\mathcal{R} = \mathbb F_4[x]/(x^2)$ where $\mathbb F_4 = \mathbb F_2(\alpha)$ with $\alpha^2 = \alpha + 1$. The code $\mathcal{C}^{(1)}$ generated by $G_1 \in \mathcal{R}^{3 \times 6}$ is MHD but not in $\langle \gamma^{s-1}\rangle^6$ and the code $\mathcal{C}^{(2)}$ generated by $G_2 \in \mathcal{R}^{3 \times 6}$ is not MHD while $\mathcal{C}^{(2)}_0$ is MDS, where
\begin{align}
G_1 :=
\begin{bmatrix}
1 & \alpha & 1 & \alpha & \alpha & 1\\
0 & x & 0 & x & (\alpha + 1)x & \alpha x\\
0 & 0 & x & x & x & x
\end{bmatrix},\quad
G_2 :=
\begin{bmatrix}
1 & 1 & 1 & 1 & 0 & 0\\
0 & x & \alpha x & (\alpha + 1)x & x & 0\\
0 & x & (\alpha + 1)x & \alpha x & 0 & x
\end{bmatrix}.
\end{align}
\end{proposition}
\begin{proof}
We have $\gamma =x, s = 2$ and $\mathcal{R} / \langle x \rangle= \mathbb F_4$. The first row of $G_1$ is already a codeword not in $(x)^6$. To prove the MHD property, we need to show $d_{\text{Hom}}(\mathcal{C}^{(1)}) > 4$. We can multiply the first row of $G_1$ by $x$ to get the generator matrix of $\mathcal{C}^{(1)}_0$. Thus, the socle can be viewed as a code $\overline{\mathcal{C}^{(1)}} \subseteq \mathbb{F}_4^6$ generated by
\begin{align} \begin{bmatrix}
1 & \alpha & 1 & \alpha & \alpha & 1\\
0 & 1 & 0 & 1 & (\alpha + 1) & \alpha\\
0 & 0 & 1 & 1 & 1 & 1
\end{bmatrix}.
\end{align}
Multiplying with an invertible matrix from the left corresponds to a change of basis and thus leaves the generated code unchanged. By multiplication with
\[\begin{bmatrix}
1 & \alpha + 1 & 0\\
0 & 1 & \alpha\\
0 & 1 & \alpha + 1
\end{bmatrix},\]
we get the matrix $G$ from \Cref{exceptionalMDS}, which means that $\overline{\mathcal{C}^{(1)}}$ is an MDS code. Thus, for all $c \in \mathcal{C}^{(1)}_0 \setminus \{0\}$ it holds $\text{wt}_{\text{Hom}}(c) = \frac{4}{3} \text{wt}_H(c) \geq \frac{4}{3} \cdot 4 = \frac{16}{3}$ (and there is a codeword which achieves equality). So to show that $d_{\text{Hom}}(\mathcal{C}^{(1)}) > 4$, it only remains to consider $c \in \mathcal{C}^{(1)} \setminus \mathcal{C}^{(1)}_0$. We must have an $\mathbb F_4^\times$-multiple of the first generator in the linear combination of $c$, so  we get $\text{wt}_{\text{Hom}}(c) = \text{wt}_H(c) = 6$. This proves that $d_{\text{Hom}}(\mathcal{C}^{(1)}) = \frac{16}{3}$, and hence $\mathcal{C}$ is MHD.

By the same argument as above, we obtain that $\mathcal{C}^{(2)}_0$ is the code from \Cref{exceptionalMDS}. That  $\mathcal{C}^{(2)}$ is not MHD follows from the first generator having homogeneous weight 4, which already implies $d_{\text{Hom}}(\mathcal{C}^{(2)}) \leq 4$.
\end{proof}

\begin{proposition}
Let $\mathcal{R} = \mathbb{Z}/4\mathbb{Z}$. The code $\mathcal{C}^{(3)}$ generated by $G_3 \in \mathcal{R}^{3 \times 4}$ is MHD but not in $\langle \gamma^{s-1} \rangle^6$ and $\mathcal{C}^{(4)}$ generated by $G_4 \in \mathcal{R}^{3 \times 4}$ is not MHD while $\mathcal{C}^{(4)}_0$ is MDS, where
\[
G_3 := 
\begin{bmatrix}
1 & 1 & 1 & 1\\
0 & 2 & 0 & 2\\
0 & 0 & 2 & 2
\end{bmatrix},\quad
G_4 := 
\begin{bmatrix}
1 & 1 & 1 & 1\\
0 & 1 & 0 & 1\\
0 & 0 & 2 & 2
\end{bmatrix}.
\]
\end{proposition}
\begin{proof}
By our usual method, we get that the socle of both codes is the single parity check code, which is MDS. Thus, $\mathcal{C}^{(3)}_0$ is MHD. Any codeword $c \in \mathcal{C}^{(3)} \setminus \mathcal{C}^{(3)}_0$ must have an $\mathcal{R}^\times$-multiple of the first row of $G_3$ in its linear combination, i.e., $c$ has all nonzero entries. So $\text{wt}_{\text{Hom}}(c) \geq 4$ and $\mathcal{C}^{(3)}$ is MHD. That $\mathcal{C}^{(4)}$ is not MHD follows from the second row of $G_4$ having homogeneous weight 2.
\end{proof}

Finally, we turn our attention to whether the MHD property is invariant under duality, as the MDS property is over finite fields. For this we go through our different cases of the characterization of MHD codes of Theorem \ref{thm:fullchar}, namely 
$K \in \{0,1,n-1,n\},$ or $n=q+1$ with $K=2$, or the lifted MDS codes for $n \leq q+1$.

If $K=0$, the dual code is $\mathcal{R}^n$, which is an MHD code. However, if $K=n,$ we do not necessarily get that $\mathcal{C}^\perp$ is MHD. In fact, we can consider the following example.

 \begin{example}
     Let us consider $\mathcal{R}= \mathbb{Z}/27\mathbb{Z}$ and the code $\mathcal{C} = \langle (1,6), (0,9) \rangle$. 
   As the code has rank $K=n=2,$ we immediately get that $\mathcal{C}$ is MHD. On the other hand, its dual $\mathcal{C}^\perp$ has rank 1 and is generated by 
     $(9,3)$ and is hence not MHD. 
 \end{example}

Also for the case $K=n-1,$ we note that again, not all dual codes have to be MHD as well.

\begin{example}
    Let us consider $\mathcal{R}= \mathbb{Z}/9\mathbb{Z}$ and the code $\mathcal{C} = \langle (1,0,1), (0,1,1) \rangle$. As $d_{\text{Hom}}(\mathcal{C})=2 >M =3/2$, we get that $\mathcal{C}$ is MHD. However, its dual code $\mathcal{C}^\perp$ is generated by 
$(8,8,1)$ and as $K=1$ but $\mathcal{C}^\perp$ is not in the socle, we get that $\mathcal{C}^\perp$ is not MHD.
\end{example}

If $K=1$, recall that $\mathcal{C}$ is only MHD if $\mathcal{C}$ has subtype $(0, \ldots, 0,1)$, which implies that $\mathcal{C}^\perp$ has rank $n.$ Thus, the dual of such MHD codes are again MHD.

If $n=q+1,$ and $ K=2$ we can apply Theorem \ref{Keqtwo} and get that $\mathcal{C}$ has subtype $(0, \ldots, 0, 2)$, and thus, the dual code $\mathcal{C}^\perp$ has rank $n$, which is an MHD code.  

For the case $n<q+1,$ we can consider again the Example \ref{ex:tighterSing}. 
  
\begin{example}
    Let $p \geq n$ and $\mathcal{C} \subseteq (\mathbb{Z}/p^s\mathbb{Z})^n$ with generator matrix
    \begin{align}
        G = \begin{bmatrix}
            1 & 0 & 1 & \cdots & 1 \\ 0 & 1 & 2 & \cdots & n-1
        \end{bmatrix}.
    \end{align}
    As the code $\mathcal{C}$ is free of rank $K=2,$ we get immediately the dual code $\mathcal{C}^\perp$ of rank $n-2$ generated by 
    \begin{align}
                     H= \begin{bmatrix}
            p^s-1 & p^s-2 & &    &  \\
            \vdots & \vdots &  & I_{n-2} &\\
            p^s-1 & p^s-n+1 &  &  & 
        \end{bmatrix}.
    \end{align}
    It can be checked that the socle of $\mathcal{C}^\perp$ is an MDS code. So for $p > 3$ we can apply Theorem \ref{mainchar} and obtain that $\mathcal{C}^\perp$ is MHD. For $n = p = 3$, $\mathcal{C}^\perp$ is of rank $1$ and not MHD. If $p = 2$ we necessarily have $n = 2$, so $\mathcal{C}^\perp = \{0\}$ is MHD trivially.
\end{example}

Lastly, we consider again the exceptional MHD code from Proposition \ref{counterex1}. That is: we set $\mathcal{R} = \mathbb F_4[x]/(x^2)$ where $\mathbb F_4 = \mathbb F_2(\alpha)$ with $\alpha^2 = \alpha + 1$. We consider the code $\mathcal{C}^{(1)}$ generated by
\begin{align}
G_1 :=
\begin{bmatrix}
1 & \alpha & 1 & \alpha & \alpha & 1\\
0 & x & 0 & x & (\alpha + 1)x & \alpha x\\
0 & 0 & x & x & x & x
\end{bmatrix}.
\end{align}

We can compute a parity-check matrix given by
\begin{align} H_1 = \begin{bmatrix}
(\alpha+1)x+1 & \alpha x & 1 & 0 & 0 & 1 \\
x+ \alpha+1 & (\alpha+1)x & 1 & 0 & 1 &  0 \\
    \alpha x + \alpha+1 & x & 1 & 1 & 0 & 0\\
     x & 0 & x & 0 & 0 & 0 \\
     x & x & 0 & 0 &  0 & 0 
\end{bmatrix}.
\end{align}
The code generated by $H_1$ has minimum homogeneous distance $d=4/3 \leq M(n-K)= 4/3,$ hence it is not an MHD code.

Finally, we might ask whether MHD codes are dense or sparse. Due to Theorem \ref{mhdismds} and \ref{thm:fullchar}, we get that the density of MHD codes is the same as that of MDS codes.
For this recall that if we fix $n$ and $K$ and let $q$ grow, then MDS codes are dense, while if we fix $q$ and a rate $R=K/n$ and let $n$ grow, MDS codes are sparse.

Thus, if we fix $q$ and $R$ and let $n$ grow, we get that the density of MDS code $\overline{\mathcal{C}} \subseteq \mathbb{F}_q^n$ is 0, and hence also the density of MHD codes over $\mathcal{R}^n,$ with residue field size $q$, which must be lifted MDS codes from $\mathbb{F}_q^n$, is 0.

On the other hand, if we fix $n,K$ and let $q$ grow, in all the cases of Theorem \ref{thm:fullchar} where $\mathcal{C} \cap \langle \gamma^{s-1} \rangle$ being MDS is equivalent to the MHD property, since MDS codes have density 1, MHD codes do as well.
 
\section{Plotkin Bound and Constant Homogeneous Weight Codes}\label{sec:plotkin}

The Plotkin bound is a standard result for codes over finite fields and has already seen a generalization for finite chain rings and the homogeneous weight. In \cite{greferath}, the authors developed such a generalization for Frobenius rings, which encompass finite chain rings. However, their proof relies on a representation of the homogeneous weight using character theory. Our goal in this section will be to prove the linear version of the bound for finite chain rings using only elementary algebra. For this purpose, we can build on the standard proof of the classical version of the bound, which proceeds via the average weight of a code. For our asymptotic analysis in Section \ref{sec:asymptotic} we will still rely on the existing result, which extends to non-linear codes.
 
\begin{definition}
The average homogeneous weight of a code is defined as
\[\overline{\text{wt}_{\text{Hom}}}(\mathcal{C}) = \frac{1}{|\mathcal{C}|}\sum_{c \in \mathcal{C}} \text{wt}_{\text{Hom}}(c).\]
\end{definition}

\begin{lemma}\label{avgweight}
Let $\mathcal{C} \subseteq \mathcal{R}^n$ be a linear code. Then $ \overline{\text{wt}_{\text{Hom}}}( \mathcal{C}) = n.$
\end{lemma}

\begin{proof}
For $i \in \{1, \ldots, n\}$, let $\pi_i:\mathcal{C} \rightarrow \mathcal{R}$ denote the projection on the $i$-th coordinate, denote its image by $I_i := \pi_i(\mathcal{C})$ and its fibers by $N_i(r):= \pi_i^{-1}(\{r\})$. Note that $\pi_i$ is a (left) $\mathcal{R}$-module homomorphism so $I_i$ is an ideal in $\mathcal{R}$.

At first, we will see that for fixed $i$, the fiber size $|N_i(r)|$ is constant for all $r \in I_i$. We show that $N_i(r) = \text{ker}(\pi_i) + c_r$ for some fixed $c_r \in N_i(r)$. The inclusion "$\supseteq$" is clear. For the other inclusion we take some $c \in N_i(r)$, write it as $c = c - c_r + c_r$ and see that $c-c_r \in \text{ker}(\pi_i)$. It follows that $|N_i(r)| = |\text{ker}(\pi_i)|$ for all $r \in I_i$ and $|N_i(r)| = 0$ for $r \notin I_i$. Additionally, notice that we can write $\mathcal{C}$ as the disjoint union $\mathcal{C} = \biguplus_{r \in \mathcal{R}} N_i(r) = \biguplus_{r \in I_i} N_i(r)$, so $|\mathcal{C}| = |\text{ker}(\pi_i)| \cdot |I_i|$.

Now we can calculate the average weight:
\begin{align}
\overline{\text{wt}_{\text{Hom}}}(\mathcal{C}) &= \frac{1}{|\mathcal{C}|} \sum_{c \in \mathcal{C}} \text{wt}_{\text{Hom}}(c) = \frac{1}{|\mathcal{C}|} \sum_{i=1}^n \sum_{c \in \mathcal{C}} \text{wt}_{\text{Hom}}(c_i)\\
&= \frac{1}{|\mathcal{C}|} \sum_{i=1}^n \sum_{r \in \mathcal{R}} \sum_{c \in N_i(r)} \text{wt}_{\text{Hom}}(c_i) = \frac{1}{|\mathcal{C}|} \sum_{i=1}^n \sum_{r \in I_i} |\text{ker}(\pi)| \text{wt}_{\text{Hom}}(r)\\ &  = \sum_{i=1}^n \frac{1}{|I_i|} \sum_{r \in I_i} \text{wt}_{\text{Hom}}(r) = n,
\end{align}
where we used \Cref{avghom} in the last equality.
\end{proof}

\begin{definition}
For an arbitrary subset $\mathcal{C} \subseteq \mathcal{R}^n$ the minimum homogeneous distance is defined as
\begin{align} d_{\text{Hom}}(\mathcal{C}) = \min \{\text{wt}_{\text{Hom}}(x-y) \mid x,y \in \mathcal{C}, x \neq y\}.\end{align}
\end{definition}

\begin{theorem}\label{plotkin2}
For an arbitrary subset $\mathcal{C} \subseteq \mathcal{R}^n$ with minimum homogeneous distance $d$, it holds that
\[
d \leq \frac{|\mathcal{C}|}{|\mathcal{C}|-1}n.
\]
If $\mathcal{C}$ is a linear code, equality is achieved if and only if $\mathcal{C}$ has constant weight, i.e., $\text{wt}_{\text{Hom}}(c) = d$ for all $c \in \mathcal{C}\setminus\{0\}$.
\end{theorem}

\begin{proof}
We only give a proof for the case that $\mathcal{C}$ is a linear code and for the additional claim. For a proof of the general case, see \cite[Proposition 2.1]{greferath}.  
It holds that 
\begin{align} 
|\mathcal{C}| \cdot \overline{\text{wt}_{\text{Hom}}}(\mathcal{C}) = \sum_{c \in \mathcal{C}} \text{wt}_{\text{Hom}}(c) = \sum_{c \in \mathcal{C}\setminus\{0\}}\text{wt}_{\text{Hom}}(c) \geq (|\mathcal{C}|-1) \cdot d,
\end{align}
where the inequality becomes an equality precisely when $\text{wt}_{\text{Hom}}$ is constant on $\mathcal{C} \setminus \{0\}$. The theorem follows by reordering and applying \Cref{avgweight}.
\end{proof}

Note that the classical Plotkin bound for the Hamming weight over finite fields is a consequence of \Cref{plotkin2}. Now we will derive a second corollary, which provides a slightly looser bound for the homogeneous weight but is a bit handier in calculation.

\begin{corollary}\label{plotkin1}
For a linear code $\mathcal{C}$ over $\mathcal{R}$ with minimum homogeneous distance $d$ it holds that
\[
d \leq \frac{q^K}{q^K-1}n.
\]
\end{corollary}
\begin{proof}
We have 
\[d(\mathcal{C}) \leq d(\mathcal{C}_0) = \frac{q}{q-1}d_H(\mathcal{C}_0) \leq \frac{q}{q-1} \frac{nq^{K-1}(q-1)}{q^K-1} = \frac{q^K}{q^K-1}n,
\]
where we applied the classical Plotkin bound to $d_H(\mathcal{C}_0)$.
\end{proof}

\subsection{Constant weight codes}

As in Section \ref{sec:singleton}, the next step will be to examine the structure of optimal codes with respect to the Plotkin bound in Theorem \ref{plotkin2}, i.e., constant weight codes. Our analysis will rely on a result by Wood that is only given for commutative chain rings. Therefore, in this Section $\mathcal{R}$ will always denote a finite \textit{commutative} chain ring.\footnote{In fact, the only point where \cite{wood} actually relies on the commutativity of $\mathcal{R}$ is in using a version of the MacWilliams extension property of the homogeneous weight. Thus, if one finds that this property also holds in the non-commutative case, our restriction will turn out to be unnecessary.}

\begin{definition}
Two codes $\mathcal{C}, \mathcal{C}'$ over $\mathcal{R}$ are equivalent if there is a vector of units $(u_1, \dots,u_n) \in (\mathcal{R}^\times)^n$ and a permutation $\sigma \in S_n$ such that 
\[\mathcal{C}' = \{(c_{\sigma(1)}u_{\sigma(1)}, \dots, c_{\sigma(n)}u_{\sigma(n)}) \mid c \in \mathcal{C}\}.\]
\end{definition}

\begin{definition}
For a positive integer $\ell$, we say that a code $\mathcal{C}'$ is an $\ell$-fold replication of $\mathcal{C}$ if   $\mathcal{C}' = \{(\underbrace{c, \dots, c}_{\ell \;\text{times}}) \mid c \in \mathcal{C}\}$.
\end{definition}

Similar to the classical case, we will say that a code $\mathcal{C}$ is non-degenerate, if for each coordinate $i \in \{1,\dots,n\}$ there exists a codeword $c \in \mathcal{C}$ with $c_i \neq 0$.

The characterization by Wood works in two steps. First, he shows that if a constant weight code over a given module exists, it is essentially unique, and second, he verifies the existence by an example. We will go the same way and adapt the respective theorems to our special case of the homogeneous weight on finite chain rings.

\begin{proposition}[\cite{wood}, adaptation of Theorem 5.4]\label{woodtheo}
Let $\mathcal{R}$ be a finite commutative chain ring with nilpotency index $s$ and let a linear code of constant homogeneous weight exist for a fixed subtype $(k_0, \dots, k_{s-1})$. Then, the constant homogeneous weight code $\mathcal{C}$ of minimal length with this subtype is unique up to equivalence. Moreover, any other non-degenerate linear code of constant homogeneous weight with this subtype is, up to equivalence, an $\ell$-fold replication of $\mathcal{C}$.
\end{proposition}

We will first construct some long code of constant weight and then examine of which minimal length constant weight code it is a replication. By the above theorem, the resulting code will be unique up to equivalence.

\begin{proposition}[\cite{wood}, Theorem 7.1]\label{longcc}
Fix a subtype $(k_0, \dots, k_{s-1})$ and denote by $V$ the $\mathcal{R}$-module $\mathcal{R}^{k_0} \times \langle \gamma \rangle^{k_1} \times \dots \times \langle \gamma^{s-1}\rangle^{k_{s-1}}$. Let $G$ be a matrix that has all nonzero vectors of $V$ as columns. Then the code $\mathcal{C}$ generated by $G$ has constant weight $|V|$, length $|V|-1$ and is of the chosen subtype.
\end{proposition}

\begin{proof}
To see that $\mathcal{C}$ is indeed of the subtype $(k_0, \dots, k_{s-1})$, we note that $G$ is, up to permutation of columns, already in systematic form, since the columns that form the necessary identity matrices all appear in $V$. We have that the rank $K$ of $\mathcal{C}$ is given by $K = \sum_{i=0}^{s-1} k_i$. For $x \in \mathcal{R}^K$ define the $\mathcal{R}$-linear homomorphism $\varphi_x: V \rightarrow \mathcal{R},\, v \mapsto \sum_{i=1}^K v_i x_i $.
By the same argument as in the proof of \Cref{avgweight}, every element of $\text{im}(\varphi_x)$ is hit $|\text{ker}(\varphi_x)|$ many times and $|V| = |\text{ker}(\varphi_x)| \cdot |\text{im}(\varphi_x)|$. Also, $\text{im}(\check x)$ is an ideal in $\mathcal{R}$.\newline
Let $c \in \mathcal{C} \setminus \{0\}$, so there exists an $x \in \mathcal{R}^K \setminus \{0\}$ with $xG = c$ and $\text{im}(\varphi_x) \neq \{0\}$. Using \Cref{avghom}, we obtain
\[
\text{wt}_{\text{Hom}}(c) 
= \sum_{v \in V} \text{wt}_{\text{Hom}}(\varphi_x(v)) 
= |\text{ker}(\varphi_x)| \sum_{r \in \text{im}(\varphi_x)} \text{wt}_{\text{Hom}}(r)
= \frac{|V|}{|\text{im}(\varphi_x)|} \sum_{r \in \text{im}(\varphi_x)} \text{wt}_{\text{Hom}}(r)
= |V|.
\]
Since $c$ was chosen arbitrarily, $\mathcal{C}$ is of constant weight $|V|$. That $\mathcal{C}$ has length $|V|-1$ is clear by the construction of $G$.
\end{proof}

The group $\mathcal{R}^\times$ acts on $V$ by scalar multiplication\footnote{In the non-commutative setting one would use right scalar multiplication.}, thereby decomposing $V$ into orbits. Denote the $\mathcal{R}^\times$-orbit of $v \in V$ by $\text{orb}(v)$. To find the desired minimal length code, it is enough to compute the gcd $g$ of the sizes of all non-zero $\mathcal{R}^\times$ orbits $\mathfrak O \subset V$. Its length will then be $\frac{|V|-1}{g}$ and its generator matrix will contain $\frac{|\mathfrak O|}{g}$ elements from each orbit as columns.

\begin{lemma}\label{orbsize}
Let $V=\mathcal{R}^{k_0} \times \langle \gamma \rangle^{k_1} \times \dots \times \langle \gamma^{s-1}\rangle^{k_{s-1}}$ and $v \in V\setminus\{0\}$.  Let $t:=\min\{i \mid \exists j: \mathfrak h(v_j) = i\}.$ Then \[|\text{orb}(v)| = q^{s-t-1}(q-1).\]
\end{lemma}

\begin{proof}
Since $v \neq 0$, we have $t < s$. Choose $j$ such that $\mathfrak h(v_j) = t$, i.e., $v_j \in \langle \gamma^t  \rangle\setminus \langle \gamma^{t+1} \rangle$. Note that for each $x \in \langle \gamma^t \rangle \setminus \langle \gamma^{t+1} \rangle$ there is some $w \in \text{orb}(v)$ with $w_j = x$. Thus, $|\text{orb}(v)| \geq |\langle \gamma^t \rangle \setminus \langle \gamma^{t+1} \rangle|$. On the other hand, let $\widetilde{w}, w' \in \text{orb}(v)$ with $\widetilde{w}_j = w_j'$. We can write $\widetilde{w} = \widetilde{u}v,  w' = u' v$ and $v_j = u\theta^t$ with $u,\widetilde{u},  u' \in \mathcal{R}^\times$, so $\widetilde{u} u \theta^t =  u' u \theta^t$. Thus, $\widetilde{u} - u' \in \langle \gamma^{s-t} \rangle$. By the minimality of $t$ and  $\widetilde{w}_i = w'_i$ for all $i \in \{1,\dots,K\}$, so $\widetilde{w} = w'$. This implies $|\text{orb}(v)| \leq |\langle \gamma^t \rangle \setminus \langle \gamma^{t+1} \rangle|$.
Finally, $|\langle \gamma^t \rangle\setminus \langle \gamma^{t+1} \rangle| = q^{s-t} - q^{s-t-1}$.
\end{proof}

\begin{corollary}
Let $V=\mathcal{R}^{k_0} \times \langle \gamma \rangle^{k_1} \times \dots \times \langle \gamma^{s-1}\rangle^{k_{s-1}}$ and $v \in V\setminus\{0\}$.   It holds that
\[\underset{v \in V \setminus \{0\}}{\gcd} |\text{orb}(v)| = (q-1).\]
\end{corollary}

\begin{proof} Let us consider again $t:=\min\{i \mid \exists j: \mathfrak h(v_j) = i\}.$
We have $\langle \gamma^{s-1}\rangle^K \subseteq V$ so there exists $v \in V$ with $t = s-1$, i.e., $|\text{orb}(v)| = q-1$. On the other hand, $q-1$ is a divisor of $|\text{orb}(w)|$ for all non-zero $w \in V$.
\end{proof}

We now summarize the results of this section in the following theorem:

\begin{theorem}\label{constweighttheo}
Let $\mathcal{R}$ be a finite commutative chain ring with maximal ideal $\langle \gamma \rangle$ and nilpotency index $s$. For each subtype $(k_0, \dots, k_{s-1})$ there exists a constant weight code $\mathcal{C}$ of minimal length which is unique up to equivalence. The generator matrix of $\mathcal{C}$ contains exactly $\frac{|\mathfrak O|}{q-1}$ elements from each non-zero $\mathcal{R}^\times$ orbit $\mathfrak O \subset V := \mathcal{R}^{k_0} \times \langle \gamma \rangle^{k_1}\times \dots \times \langle \gamma^{s-1} \rangle^{k_{s-1}}$ as columns. $\mathcal{C}$ has length $\frac{|V|-1}{q-1}$ and constant weight $\frac{|V|}{q-1}$ where $|V| = \prod_{i=0}^{s-1}q^{k_i (s-i)}$. Each constant weight code of the same subtype is up to equivalence an $\ell$-fold replication of $\mathcal{C}$.
\end{theorem}

\begin{proof}
The $(q-1)$-fold replication of the described code $\mathcal{C}$ is equivalent to the code from \Cref{longcc}. Since equivalent codes have the same weight distribution, $\mathcal{C}$ has constant weight $\frac{|V|}{q-1}$. As $\langle \gamma^{s-1} \rangle^K \subseteq V$, there exists an orbit $\mathfrak O \subset V$ with $|\mathfrak O| = q-1$. Thus, the generator matrix of $\mathcal{C}$ contains only one element of $\mathfrak O$ as a column and $\mathcal{C}$ is not equivalent to a replication of a shorter code. The remaining claims follow by \Cref{woodtheo}.
\end{proof}

\begin{example}\label{constweightex}
Let $\mathcal{R} = \mathbb Z/4\mathbb Z$ and choose the subtype $(1,1)$. One representation $\mathcal{C}$ of the code described in \Cref{constweighttheo} is generated by
\[\begin{pmatrix}
0 & 1 & 1 & 2 & 2 & 3 & 3\\
2 & 0 & 2 & 0 & 2 & 0 & 2
\end{pmatrix}.\]
We have $q = 2$ and $s = 2$, so $|V| = 2^3 = 8$ and $\mathcal{C}$ accordingly has length $7$ and constant weight $8$.
\end{example}
 
\section{Asymptotic Bounds}\label{sec:asymptotic}
In this section we want to analyze the asymptotic behavior of both the Singleton and the Plotkin bound. Here, we will turn around the direction of the inequalities and bound from above the relative size, i.e., the rate of the code in terms of its length and minimum distance. 

\begin{definition}\label{asymdef}
For a positive integer $n$ and  a non-negative integer $d$ denote by $B_\mathcal{R}(n,d)$ the largest number of codewords in a linear code of length $n$ and homogeneous distance at least $d$, i.e., \[B_\mathcal{R}(n,d) := \max \,\{|\mathcal{C}| \mid \mathcal{C} \subseteq \mathcal{R}^n \text{ linear code},\, d(\mathcal{C}) \geq d\},\]
where we set the maximum to be 0 if the set is empty.

For $\delta \geq 0$, the asymptotic rate is defined as
\[\beta_\mathcal{R}(\delta) := \limsup_{n \rightarrow \infty} \frac{1}{n} \log_{|\mathcal{R}|} (B_\mathcal{R}(n,\delta n)).\]
\end{definition}
 
From the Singleton bound that involves the $\mathcal{R}$-dimension of the code, we can immediately derive an asymptotic version: Take $\mathcal{\mathcal{C}}$ to be a linear code over $\mathcal{R}$ of length $n$ with $|\mathcal{\mathcal{C}}| = B_{\mathcal{R}}(n, \delta n)$ and $d_{\text{Hom}}(\mathcal{C}) \geq \delta n$. Then by \Cref{singleton2}, 
\[\log_{|\mathcal{R}|}(|\mathcal{C}|) \leq n - \left\lfloor \frac{d_{\text{Hom}}(\mathcal{C})(q-1)}{q}\right\rfloor_h \leq n - \left\lfloor \frac{\delta n(q-1)}{q}\right\rfloor_h,\] 
and thus after dividing by $n$ and taking the limit we obtain
\[\beta_\mathcal{R}(\delta) \leq 1 - \frac{q-1}{q}\delta.\]
 
The Plotkin bound leads to a tighter asymptotic bound than the Singleton bound.

\begin{theorem}
It holds $\beta_\mathcal{R}(\delta) = 0$ if $1 \leq \delta$, and $\beta_\mathcal{R}(\delta) \leq 1 - \delta$ if $0 \leq \delta \leq 1$.
\end{theorem}

\begin{proof}
First, for $d > n$ we rearrange the Plotkin bound from \Cref{plotkin2} to $|\mathcal{C}| \leq \frac{d}{d-n}$. So if $\delta > 1$ and $\mathcal{C}$ is a linear code that achieves $B_\mathcal{R}(n,\delta n)$, we get $|\mathcal{C}| \leq \frac{\delta n}{\delta n - n}$. It follows that
\begin{align}
\beta_\mathcal{R}(\delta) &\leq \limsup_{n \rightarrow \infty} \frac{1}{n} \log_{|\mathcal{R}|}\left(\frac{\delta n}{\delta n - n}\right)
= \limsup_{n \rightarrow \infty} \frac{1}{n} \log_{|\mathcal{R}|}(\delta n) - \frac{1}{n} \log_{|\mathcal{R}|}((\delta - 1)n) = 0.
\end{align}
Now assume that $0 \leq \delta \leq 1$ and take a linear code $\mathcal{C}$ that achieves $B_\mathcal{R}(n, \delta n)$. Let $n':= \left\lfloor \delta n - 1\right\rfloor < n$. By a pigeonhole argument, there is at least one $(n - n')$-tuple of elements from $\mathcal{R}$ such that the number of codewords in $\mathcal{C}$ starting with this tuple is at least $\frac{|\mathcal{C}|}{|\mathcal{R}|^{n - n'}}$. For this subset of $\mathcal{C}$, delete the first coordinates to obtain the set $\mathcal{C}' \subseteq \mathcal{R}^{n'}$. By our construction, $|\mathcal{C}'| \geq \frac{|\mathcal{C}|}{|\mathcal{R}|^{n - n'}}$ and $d(\mathcal{C}') \geq d(\mathcal{C}) \geq \delta n$. It follows that $n' < d(\mathcal{C}')$ so we can apply our rearranged Plotkin bound to $\mathcal{C}'$ and get 
\[
\frac{|\mathcal{C}|}{|\mathcal{R}|^{n-n'}} \leq |\mathcal{C}'| \leq \frac{d(\mathcal{C}')}{d(\mathcal{C}') - n'} \leq \frac{\delta n}{\delta n - n'} \leq \delta n,
\]
where we have used that $\delta n - n' \geq 1$ and the general fact that for $a > b > c > 0$ it holds $\frac{a}{a-c} < \frac{b}{b-c}.$
Therefore, $|\mathcal{C}| \leq |\mathcal{R}|^{n-n'}\delta n$ and asymptotically
\begin{align}
\beta_\mathcal{R}(\delta) &\leq \limsup_{n \rightarrow \infty} \frac{1}{n} \log_{|\mathcal{R}|}(|\mathcal{R}|^{n-n'}\delta n)  = \limsup_{n \rightarrow \infty} \frac{1}{n} (n - n' + \log_{|\mathcal{R}|}(\delta n)) = 1 - \lim_{n \rightarrow \infty} \frac{n'}{n} = 1 - \delta.
\end{align}
\end{proof}
 
\begin{figure}[ht]
\begin{subfigure}{0.48\textwidth}
\begin{tikzpicture}
\begin{axis}[
    axis lines = left,
    xlabel={$\delta$},
    ylabel={$\beta_\mathcal{R}(\delta)$},
    every axis x label/.style={
        at={(current axis.right of origin)},
        anchor=west
    },
    every axis y label/.style={
        at={(current axis.above origin)},
        anchor=south
    },
    xmin=0, xmax=1.6,
    ymin=0, ymax=1.1,
    xtick={0,1,1.5},
    xticklabels={$0$, $1$, $\frac{3}{2}$},
    ytick={0,1},
    yticklabels={$0$, $1$},
    axis equal image,
]
 
\addplot[red, thick, domain=0:1.5] {1 - (2/3)*x};
\node[red] at (axis cs:1.25,0.35) {Singleton};
 
\addplot[green!50!black, thick, domain=0:1] {1 - x};
\node[green!50!black] at (axis cs:0.55,0.25) {Plotkin};

\end{axis}
\end{tikzpicture}
\caption{Homogeneous metric}
\end{subfigure}
\hfill
\begin{subfigure}{0.48\textwidth}
\begin{tikzpicture}
\begin{axis}[
    axis lines = left,
    xlabel={$\delta$},
    ylabel={$\beta_{\mathbb F_q}^H(\delta)$},
    every axis x label/.style={
        at={(current axis.right of origin)},
        anchor=west
    },
    every axis y label/.style={
        at={(current axis.above origin)},
        anchor=south
    },
    xmin=0, xmax=1.6,
    ymin=0, ymax=1.1,
    xtick={0,0.67,1},
    xticklabels={$0$, $\frac{2}{3}$, $1$},
    ytick={0,1},
    yticklabels={$0$, $1$},
    axis equal image,
]
 
\addplot[red, thick, domain=0:1.5] {1 - x};
\node[red] at (axis cs:0.9,0.35) {Singleton};
 
\addplot[green!50!black, thick, domain=0:1] {1 - (3/2)*x};
\node[green!50!black] at (axis cs:0.35,0.2) {Plotkin};

\end{axis}
\end{tikzpicture}
\caption{Hamming metric}
\end{subfigure}
\caption{Comparison of the asymptotic bounds in the homogeneous metric and the Hamming metric, where $q = 3$}
\label{fig:asymcomp}
\end{figure}
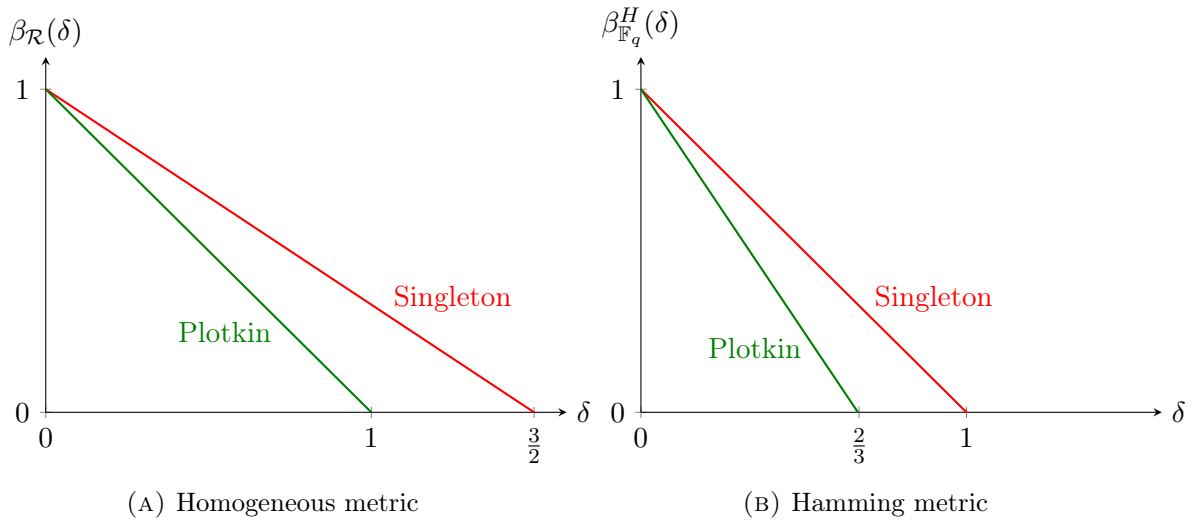

As it turns out, both the Plotkin and Singleton bound asymptotically behave like the classical bounds for a scaled Hamming metric. If we define $\beta_{\mathbb F_q}^H$ analogously to \Cref{asymdef} but with the Hamming metric, it follows  that $\beta_{\mathbb F_q}^H \leq 1 - \delta$ by the asymptotic Hamming metric Singleton bound and $\beta_{\mathbb F_q}^H \leq \max(0, 1 - \frac{q}{q-1}\delta)$ by the asymptotic Hamming metric Plotkin bound. So they only differ by a factor of $\frac{q}{q-1}$ from our results in this section. The analogy is also graphically visible in \Cref{fig:asymcomp}. That the Plotkin bound is asymptotically tighter than the Singleton bound fits together with the fact that there are no MHD or MDS codes for large $n$ by the MDS conjecture.

\section{Conclusion}

In this paper we studied optimal codes for classical bounds in the homogeneous metric over finite chain rings. 
We first revisited Singleton-type bounds in this setting and introduced the notion of Maximum Homogeneous Distance (MHD) codes. 
Our characterization shows that these codes are closely related to classical MDS codes and that their structure is largely determined by the socle of the code. 
This allowed us to analyze the density of MHD codes and to compare their behavior to the classical setting.

We then turned to the Plotkin bound for the homogeneous weight and investigated constant-weight codes attaining this bound. 
Building on Wood's structural results, we identified a missing construction of the shortest-length constant-weight homogeneous code. 
By providing this construction we completed the characterization of such codes.

These results further illustrate how changing the underlying metric can significantly alter the landscape of optimal codes. 
A natural direction for future work is to investigate whether the currently known Singleton-type bound for the homogeneous metric is the most appropriate analogue of the classical Singleton bound. 
In particular, it would be interesting to determine whether tighter bounds can be obtained, similar to the refinements that were recently achieved for the Lee metric in \cite{LeeSB}.
 
\section*{Acknowledgments}
 Violetta Weger’s work was supported by the Technical University of Munich—Institute for Advanced Study.

\bibliographystyle{plainnat}

\bibliography{biblio}

\begin{thebibliography}{14}
\providecommand{\natexlab}[1]{#1}
\providecommand{\url}[1]{\texttt{#1}}
\expandafter\ifx\csname urlstyle\endcsname\relax
  \providecommand{\doi}[1]{doi: #1}\else
  \providecommand{\doi}{doi: \begingroup \urlstyle{rm}\Url}\fi

\bibitem[Ball(2012)]{simeon}
Simeon Ball.
\newblock On sets of vectors of a finite vector space in which every subset of basis size is a basis.
\newblock \emph{Journal of the European Mathematical Society (EMS Publishing)}, 14\penalty0 (3), 2012.
\newblock \doi{10.4171/jems/316}.

\bibitem[Bariffi and Weger(2025)]{LeeSB}
Jessica Bariffi and Violetta Weger.
\newblock Better bounds on the minimal {Lee} distance.
\newblock \emph{SIAM Journal on Discrete Mathematics}, 2025.

\bibitem[Byrne and Weger(2023)]{LeePaper}
Eimear Byrne and Violetta Weger.
\newblock Bounds in the {Lee} metric and optimal codes.
\newblock \emph{Finite Fields and Their Applications}, 87:\penalty0 102151, 2023.
\newblock \doi{10.1016/j.ffa.2022.102151}.

\bibitem[Constantinescu and Heise(1997)]{constantinescu}
Ioana Constantinescu and Werner Heise.
\newblock A metric for codes over residue class rings.
\newblock \emph{Problemy Peredachi Informatsii}, 33:\penalty0 22--28, 1997.
\newblock URL \url{https://www.mathnet.ru/eng/ppi375}.

\bibitem[Greferath and O’Sullivan(2004)]{greferath}
Marcus Greferath and Michael~E. O’Sullivan.
\newblock On bounds for codes over frobenius rings under homogeneous weights.
\newblock \emph{Discrete Mathematics}, 289:\penalty0 11--24, 2004.
\newblock \doi{10.1016/j.disc.2004.10.002}.

\bibitem[Hammons et~al.(1994)Hammons, Kumar, Calderbank, Sloane, and Sole]{hammons}
A.R. Hammons, P.V. Kumar, A.R. Calderbank, N.J.A. Sloane, and P.~Sole.
\newblock The {Z}4-linearity of kerdock, preparata, goethals, and related codes.
\newblock \emph{IEEE Transactions on Information Theory}, 40:\penalty0 301--319, 1994.
\newblock \doi{10.1109/18.312154}.

\bibitem[Hurley(2019)]{hurley}
Ted Hurley.
\newblock {MDS} codes over finite fields, 2019.

\bibitem[Nechaev(1991)]{nechaev}
A.~A. Nechaev.
\newblock Kerdock code in a cyclic form.
\newblock \emph{Discrete Mathematics and Applications}, 1:\penalty0 365--384, 1991.
\newblock \doi{10.1515/dma.1991.1.4.365}.

\bibitem[Nechaev and Honold(1999)]{nechaevhonold}
A.~A. Nechaev and T.~Honold.
\newblock Weighted modules and representations of codes.
\newblock \emph{Problems of Information Transmission}, 35:\penalty0 18--39, 1999.
\newblock URL \url{https://www.mathnet.ru/eng/ppi450}.

\bibitem[Norton and Sălăgean(2000)]{norton}
Graham~H. Norton and Ana Sălăgean.
\newblock On the structure of linear and cyclic codes over a finite chain ring.
\newblock \emph{Applicable Algebra in Engineering, Communication and Computing}, 10:\penalty0 489--506, 2000.
\newblock \doi{10.1007/PL00012382}.

\bibitem[Pyka(2025)]{bachelor}
Andreas Pyka.
\newblock Optimal codes for the homogeneous weight.
\newblock Bachelor's thesis, Technical University of Munich, 2025.

\bibitem[Samei and Mahmoudi(2018)]{samei}
Karim Samei and Saadoun Mahmoudi.
\newblock Singleton bounds for r-additive codes.
\newblock \emph{Advances in Mathematics of Communications}, 12:\penalty0 107--114, 2018.
\newblock \doi{10.3934/amc.2018006}.

\bibitem[Shiromoto(2000)]{shiromoto}
Keisuke Shiromoto.
\newblock Singleton bounds for codes over finite rings.
\newblock \emph{Journal of Algebraic Combinatorics}, 12:\penalty0 95--99, 2000.
\newblock \doi{10.1023/A:1008767703006}.

\bibitem[Wood(2002)]{wood}
Jay~A. Wood.
\newblock The structure of linear codes of constant weight.
\newblock \emph{Transactions of the American Mathematical Society}, 354:\penalty0 1007--1026, 2002.
\newblock \doi{10.1090/S0002-9947-01-02905-1}.

\end{thebibliography}

\end{document}